\documentclass[smallextended]{svjour3}
\journalname{Algorithmica}
\usepackage{fullpage,graphicx}
\usepackage[numbers]{natbib}

\newcommand{\Oh}[1]
    {\ensuremath{\mathcal{O} \hspace{-.5ex} \left( {#1} \right)}}
\newcommand{\occ}
    {\ensuremath{\mathrm{occ}}}
\newcommand{\extract}
    {\ensuremath{\mathsf{extract}}}

\begin{document}

\title{Faster Approximate Pattern Matching in Compressed Repetitive Texts
    \thanks{A preliminary version of this work appeared in the proceedings of ISAAC 2011~\cite{GGP11}.}}

\author
{Travis Gagie \and
Pawe{\l} Gawrychowski \and\\
Christopher Hoobin \and
Simon J. Puglisi}

\institute
{T. Gagie \at
Department of Computer Science\\
Aalto University, Finland.\\
\email{travis.gagie@aalto.fi} \\[1ex]
\and
P. Gawrychowski \at
Max Planck Institute\\
Saarbr\"ucken, Germany\\
\email{gawry@cs.uni.wroc.pl} \\[1ex]
\and
C. Hoobin \at
School of Computer Science and Information Technology\\
Royal Melbourne Institute of Technology, Australia\\
\email{christopher.hoobin@rmit.edu.au} \\[1ex]
\and
S. J. Puglisi \at
Department of Informatics\\
King's College London, United Kingdom\\
\email{simon.puglisi@kcl.ac.uk}}

\maketitle \thispagestyle{empty}

\begin{abstract}
Motivated by the imminent growth of massive, highly redundant genomic databases,
we study the problem of compressing a string database while simultaneously
supporting fast random access, substring extraction and pattern matching to
the underlying string(s).  Bille et al. (2011) recently showed how, given a straight-line program with
$r$ rules for a string $s$ of length $n$, we can build an $\Oh{r}$-word data
structure that allows us to extract any substring of length $m$ in $\Oh{\log
n + m}$ time.  They also showed how, given a pattern $p$ of length $m$
and an edit distance \(k \leq m\), their data structure supports finding
all \occ\ approximate matches to $p$ in $s$ in $\Oh{r (\min (m k, k^4 + m) +
\log n) + \occ}$ time.  Rytter (2003) and Charikar et al. (2005) showed that
$r$ is always at least the number $z$ of phrases in the LZ77 parse of $s$,
and gave algorithms for building straight-line programs with $\Oh{z \log n}$
rules. In this paper we give a simple $\Oh{z \log n}$-word data structure
that takes the same time for substring extraction but only $\Oh{z \min (m
k, k^4 + m) + \occ}$ time for approximate pattern matching.
\end{abstract}

\section{Introduction} \label{sec:intro}

The recent revolution in high-throughput sequencing technology has made the
acquisition of large genomic sequences drastically cheaper and faster. As the
new technology takes hold, ambitious sequencing projects such as the 1,000
Human Genomes~\cite{Dur10} and the 10,000 Vertebrate Genomes~\cite{G10K}
projects are set to create large databases of strings (genomes) that vary only
slightly from each other, and so will contain large numbers of long repetitions.
Efficient storage of these collections is not enough: fast access to enable
search and sequence alignment is paramount. The utility of such a data structure
is not limited to the treatment of DNA collections. Ferragina and
Manzini's recent study of the compressibility of web pages reveals enormous
redundancy in web crawls~\cite{FM10}. Exploiting this redundancy to reduce
space while simultaneously enabling fast access and search over crawled pages
(for snippet generation or cached page retrieval) is a significant challenge.
The problem of compressing and indexing such highly repetitive strings (or string
collections) was introduced in~\cite{SVMN08} (see also~\cite{MNSV10}).
With an LZ78- or BWT-based data
structure~\cite{ANS12,FV07} we can store a string $s$ of length $n$ in space bounded in terms of the
$t$th-order empirical entropy~\cite{Man01}, for any \(t = o (\log_\sigma n)\), and later
extract any substring of length $m$ in $\Oh{m / \log_\sigma n}$ time.  For
very repetitive texts, however, compression based on the LZ77~\cite{ZL77} can
use significantly fewer than \(n H_t (s)\) bits~\cite{SVMN08}.

Rytter~\cite{Ryt03} showed that the number $z$ of phrases in the LZ77 parse of $s$ is at most the number of rules in the smallest straight-line program (SLP) for $s$\footnote{In this paper we consider only the version of LZ77 without self-referencing, sometimes called LZSS~\cite{SS82}.}.
He then showed how the LZ77 parse can be turned into an SLP for $s$ with $\Oh{z \log n}$ rules whose parse-tree has height $\Oh{\log n}$.  This SLP can be viewed as a data structure that stores $s$ in $\Oh{z \log n}$ words and supports substring extraction in $\Oh{\log n + m}$ time.  Bille, Landau, Raman, Rao, Sadakane and Weimann~\cite{BLRSSW11} showed how, given an SLP for $s$ with $r$ rules, we can build a data structure that takes $\Oh{r}$ words and supports substring extraction in $\Oh{\log n + m}$ time regardless of the height of the parse tree.  Unfortunately, since no polynomial-time algorithm is known to produce an SLP for $s$ with \(o (z \log n)\) rules, even with no bound on the height, we still do not know how, efficiently, to build a data structure that has better bounds than Rytter's.

Bille et al.~\cite{BLRSSW11} also show how, given a pattern $p$ of length $m$ and an edit distance \(k \leq m\), their data structure supports finding all \occ\ approximate matches to $p$ in $s$ in $\Oh{r (\min (m k, k^4 + m) + \log n) + \occ}$ time.  Their main idea is that, if there is a rule \(X \rightarrow Y Z\) in the SLP and we have already found all the approximate matches in expansions of $Y$ and $Z$ then, to find all the approximate matches in the expansion of $X$, we need only search the substring consisting of the \(m + k\) last characters of $Y$'s expansion concatenated with the first \(m + k\) characters of $Z$'s expansion.  Extracting these characters with their data structure takes $\Oh{\log n + m}$ time per rule, or $\Oh{r (\log n + m)}$ time in total.  In this paper we discuss two improvements to this idea: first, by the same argument, we need only search the \(m + k\) characters to either side of the phrase boundaries in the LZ77 parse; second, since we know in advance where those phrase boundaries are, we do not need the full power of random access.  Our first observation immediately improves Bille et al.'s time bound for approximate matching to $\Oh{z (\min (m k, k^4 + m) + \log n) + \occ}$, while our second has led us to develop a data structure whose time bound is $\Oh{z \min (m k, k^4 + m) + \occ}$.

Neither Rytter's nor Bille et al.'s data structures are practical.  However, in another strand of recent work, Kreft and Navarro~\cite{KN10,KN11} introduced a variant of LZ77 called LZ-End and gave a data structure based on it with which we can store $s$ in \(\Oh{z' \log n} + o (n)\) bits, where $z'$ is the number of phrases in the LZ-End parse of $s$, and later extract any {\em phrase} (not arbitrary substring) in time proportional to its length.  The \(o (n)\) term can be removed at the cost of slowing extraction down by an $\Oh{\log n}$ factor.  Extracting arbitrary substrings is fast in practice but could be slow in the worst case.  Also, although the LZ-End encoding is small in practice for very repetitive strings, it is not clear whether $z'$ can be bounded in terms of $z$.

\paragraph{Our Contribution.} In this paper we describe a simple $\Oh{z \log n}$-word data structure,
which we call the {\em block graph} for $s$, that takes $\Oh{\log n + \ell - f}$ time to extract any
substring \(s [f..\ell]\) but lets us add bookmarks to speed up extraction from pre-specified points.  This allows us to find all \occ\ approximate matches of a pattern of length $m$ in $\Oh{z \min (m k + m, k^4 + m) + \occ}$ time. Our space bound (in terms of $z$) and substring extraction time are the same as Bille et al.'s~\cite{BLRSSW11}; our approximate pattern matching time is faster both because we replace $r$ by $z$ (which, as noted above, they can too) and because we remove the \(\log n\) term, which is due to the overhead for random access.  More importantly, however, our results require much simpler machinery.  We believe the block graph is the first practical data structure with solid theoretical guarantees for compression and retrieval of highly repetitive collections.

In the next section we describe the block graph. Then, in Section~\ref{sec:theoretical}, we
relate the size of the block graph to the size of the LZ77 parsing of its underlying string.
We show that a block graph naturally compresses the string while allowing efficient random
access and extraction of substrings. In Section~\ref{sec:accelerated} we show how to augment
the block graph to support fast approximate pattern matching.  In Section~\ref{sec:experiments} we describe a practical implementation of the block graph and compare its performance to that of Kreft and Navarro's data structure.

We note that the idea of searching only around phrase boundaries in the LZ77 parse could be useful in other contexts.  For example, suppose we want to build an index for approximate pattern matching in a text and we know in advance reasonable upper bounds $M$ and $K$ on the lengths of the patterns and the edit distances in which we will be interested.  We can extract the \(M + K\) characters to either side of each boundary, obtaining substrings of length \(2 (M + K)\); separate each pair of consecutive substrings by \(K + 1\) copies of a character not in the alphabet; and build an index for the resulting modified string, which could be much smaller.  For any pattern of length at most $M$ and any edit distance at most $K$, the original string contains an approximate match if and only if the modified string does; moreover, from the positions of the approximate matches in the modified string and the structure of the LZ77 parse, we can use two-sided range reporting to deduce the positions of the approximate matches in the original string~\cite{GGKNP12}.  We hope to use similar ideas to reduce the space usage of hash-based indexes~\cite{VDTP12}.

\section{Block graphs} \label{sec:blocks}

For the moment, assume \(n = 2^h\) for some integer $h$.  We start building the block graph of $s$ with node \(\langle 1..n \rangle\), which we call the root and consider to be at depth 0.  For \(0 \leq d < t\), for each node \(v = \langle i..i + b - 1 \rangle\) at depth $d$, where \(b = 2^{t - d}\) is the block size at depth $d$, we add pointers from $v$ to nodes \(\langle i..i + b / 2 - 1 \rangle\), \(\langle i + b / 4..i + 3 b / 4 - 1 \rangle\) and \(\langle i + b / 2..i + b - 1 \rangle\), creating those nodes if necessary.  We call these three nodes the children of $v$ and each other's siblings, and we call $v$ their parent.  Notice that a node can have two parents.  We associate with each node \(\langle i..j \rangle\) the block \(s [i..j]\) of characters in $s$.  If $n$ is not a power of 2, then we append blanks to $s$ until it is.  After building the block graph, we remove any nodes whose blocks contain only blanks or blanks and characters in another block at the same depth, and replace any node \(\langle i..j \rangle\) with \(j > n\) by \(\langle i..n \rangle\).  We delete all pointers to any such nodes.

We can reduce the size of the block graph by truncating it such that we keep only the nodes at depths where storing three pointers takes less space than storing a block of characters explicitly.  We mark as an internal node each node whose block is the first occurrence of that substring in $s$.  At the deepest internal nodes, instead of storing a pointer, we store the nodes' blocks explicitly.
We mark as a leaf all nodes whose block is not unique and whose parents are internal nodes.
We then remove any node that is not marked as an internal node or a leaf.  Figure~\ref{fig:fibonacci} shows the block graph for the eighth Fibonacci string, {\sf abaababaabaababaababa}, truncated at depth 3.  Oval nodes are internal nodes and rectangular nodes are leaves.  Notice that the root has only two children, because the block for node \(\langle 17..32 \rangle\) would contain only blanks and characters in \(s [9..21]\), so \(\langle 17..32 \rangle\) is removed; similarly, \(\langle 21..24 \rangle\) is removed.

\begin{figure*}
\begin{center}
\includegraphics[width=.9\textwidth]{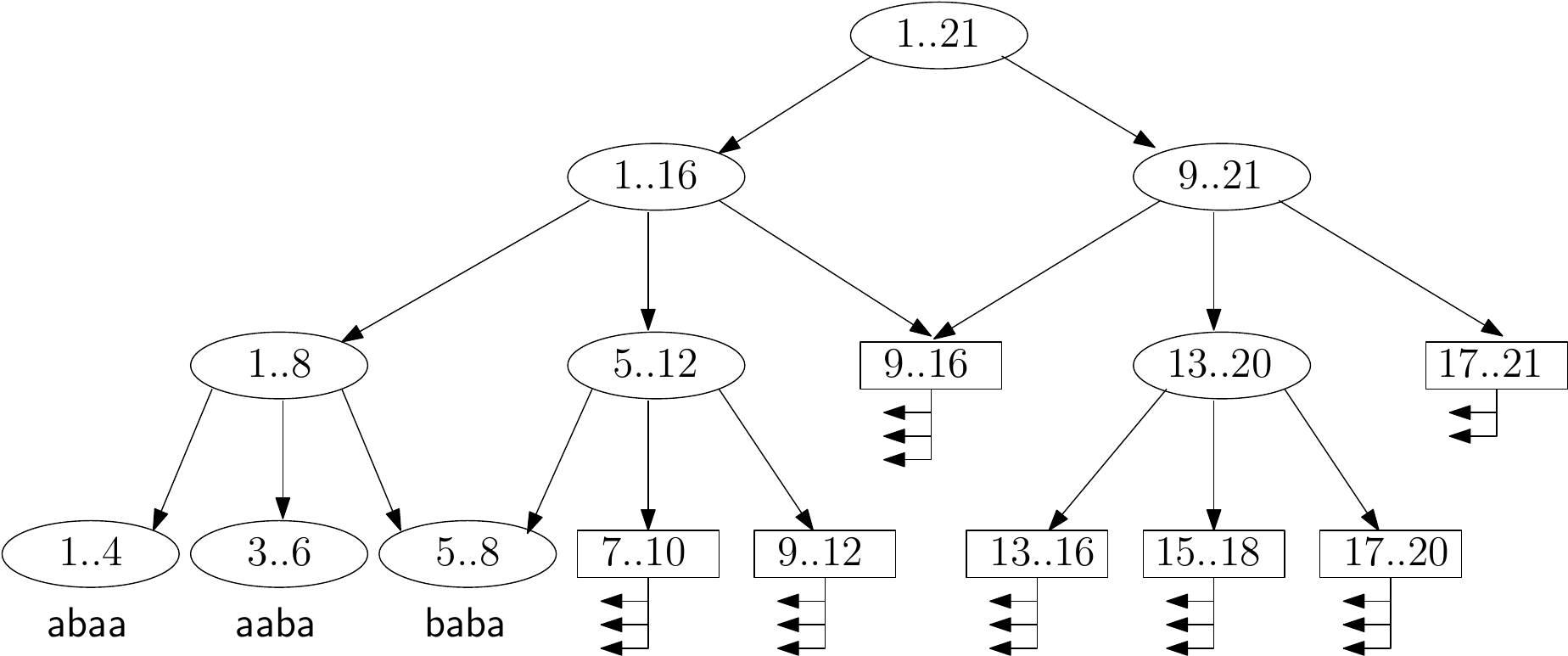}
\caption{The block graph for the eighth Fibonacci string, {\sf abaababaabaababaababa}, truncated at depth 3.}
\label{fig:fibonacci}
\end{center}
\end{figure*}

The key phase in building the block graph is updating the leaves' pointers, shown in Figure~\ref{fig:fibonacci} as the arrows below rectangular nodes.  Suppose a leaf $u$ at depth $d$ had a child \(\langle i..j \rangle\), which was been removed because it was neither an internal node nor a leaf.  Consider the first occurrence \(s [i'..j']\) in $s$ of the substring \(s [i..j]\).  Notice that \(s [i'..j']\) is completely contained within some block at depth $d$ --- this is one reason why we use overlapping blocks --- and, since \(s [i'..j']\) is the first occurrence of that substring in $s$, that block is associated with an internal node $v$.  We replace the pointer from $u$ to \(\langle i..j \rangle\) by a pointer to $v$ and the offset of $i'$ in $v$'s block.  For the example shown in Figure~\ref{fig:fibonacci}, \(\langle 17..21 \rangle\) previously had children \(\langle 17..20 \rangle\) and \(\langle 19..21 \rangle\).  The blocks \(s [17..20] = abab\) and \(s [19..21] = aba\), which first occur in positions 4 and 1, respectively.  Therefore, we replace \(\langle 17..21 \rangle\)'s pointer to \(\langle 17..20 \rangle\) by a pointer to \(\langle 1..8 \rangle\) and the offset 3; we replace its pointer to \(\langle 19..21 \rangle\) by another pointer to \(\langle 1..8 \rangle\) and the offset 0.

Extracting a single character \(s [i]\) in $\Oh{\log n}$ time is fairly straightforward: we start at the root and repeatedly descend to any child whose block contains \(s [i]\); if we come to a leaf $u$ such that \(s [i]\) is the $j$th character in $u$'s block but, instead of pointing to a child whose block contains \(s [i]\), $u$ stores a pointer to internal node $v$ and offset $c$, then we follow $u$'s pointer to $v$ and extract the \((j + c)\)th character in $v$'s block; finally, when we arrive at an internal node with maximum depth, we report the appropriate character of its block, which is stored there explicitly. By definition the maximum depth of the block graph is $\log n$ and at each depth, we either descend immediately in $\Oh{1}$ time, or follow a pointer from a leaf to an internal node in $\Oh{1}$ time and then descend.  Therefore, we use a total of $\Oh{\log n}$ time.

For example, suppose we want to extract the 11th character from \(s = \mathsf{abaababaabaababaababa}\) using the block graph shown in Figure~\ref{fig:fibonacci}.  Starting at the root, we can descend to either child, since both their blocks contain \(s [11]\); suppose we descend to the left child, \(\langle 1..16 \rangle\).  From \(\langle 1..16 \rangle\) we can descend to either the middle or right children; suppose we descend to the right child, \(\langle 9..16 \rangle\).  Since \(\langle 9..16 \rangle\) is a leaf, the pointer to child \(\langle 9..12 \rangle\) has been replaced by a pointer to \(\langle 1..8 \rangle\) and offset 0, while the pointer to child \(\langle 11..14 \rangle\) has been replaced by another pointer to \(\langle 1..8 \rangle\) and offset 2.  This is because the first occurrence of \(s [9..12] = \mathsf{abaa}\) is \(s [1..4]\) and the first occurrence of \(s [11..14] = \mathsf{aaba}\) is \(s [3..6]\). Suppose we follow the second pointer.  Since we would have extracted the first character from \(\langle 11..14 \rangle\)'s block, we are now to extract the third character from \(\langle 1..8 \rangle\)'s block.  We can descend to either \(\langle 1..4 \rangle\) and extract the third character of its block, or descend to \(\langle 3..6 \rangle\) and extract the first character of its block.

Extracting longer substrings is similar, but complicated by the fact that we want to avoid breaking the substring into too many pieces as we descend.  In the next section we will show how to extract any substring of length $m$ in $\Oh{\log n + m}$ time; however, we first prove an upper bound on the block graph's size.

\section{Fast random access in compressed space} \label{sec:theoretical}

In this section we show that block graphs achieve compression while simultaneously allowing
easy access to the underlying string. Our space result relies on the following easily proved
lemma.

\begin{lemma}[\cite{GG10}] \label{prop:breaks}
The first occurrence of any substring in $s$ must touch at least one boundary between phrases in the LZ77 parse.
\end{lemma}

Lemma~\ref{prop:breaks} allows us to relate the size of the block graph to the LZ77
parsing of the underlying string, as summarized below.

\begin{theorem}
The block graph for $s$ takes $\Oh{z \log^2 n}$ bits.
\end{theorem}

\begin{proof}
Each internal node's block is the first occurrence of that substring in $s$ so, by Proposition~\ref{prop:breaks}, it must touch at least one boundary between phrases in the LZ77 parse.  Since each such boundary can touch at most three blocks in the same level, there are at most \(3 z\) internal nodes in each level.  It follows that there are $\Oh{z \log n}$ nodes in all.  Since each node stores $\Oh{\log n}$ bits, the whole block graph takes $\Oh{z \log^2 n}$ bits.
\end{proof}

We define the query \(\extract (u, i, j)\) to return the $i$th through $j$th characters in $u$'s block.  Notice that, if $u$ is the root, then these characters are \(s [i..j]\).  We now show how to implement \extract\ queries in such a way that extracting a substring of $s$ with length $m$ takes $\Oh{\log n + m}$ time.

There are three cases to consider when performing \(\extract (u, i, j)\): $u$ could be an internal node at maximum depth, in which case we simply return the $i$th through $j$th characters of its block, which are stored explicitly; $u$ could be an internal node with children; or $u$ could be a leaf.  First suppose that $u$ is an internal node with children.  Let $d$ be $u$'s depth and \(b = 2^{\lceil \log_2 n \rceil - d}\); notice $b$ is the length of $u$'s block unless the block is a suffix of $s$, in which case the block might be shorter.  If the interval \([i..j]\) is completely contained in one of the intervals \([1..b / 2]\), \([b / 4 + 1..3 b / 4]\) or \([b / 2 + 1..b]\), then we set $v$ to be the left, middle or right child of $u$, respectively (choosing arbitrarily if two intervals each completely contain \([i..j]\)), and implement \(\extract (u, i, j)\) as either \(\extract (v, i, j)\), \(\extract (v, i - b / 4, j - b / 4)\) or \(\extract (v, i - b / 2, j - b / 2)\).  Otherwise, \([i..j]\) must be more than a quarter of \([1..b]\) and we can split \([i..j]\) into 2 or 3 subintervals, each of length at least \(b / 8\) but completely contained in one of  \([1..b / 2]\), \([b / 4 + 1..3 b / 4]\) or \([b / 2 + 1..b]\); this is the other reason why we use overlapping blocks.  We implement \(\extract (u, i, j)\) with an \extract\ query for each subinterval.

Now suppose that $u$ is a leaf.  Again, let $d$ be $u$'s depth and \(b = 2^{\lceil \log_2 n \rceil - d}\).  If the interval \([i..j]\) is completely contained in one of the intervals \([1..b / 2]\), \([b / 4 + 1..3 b / 4]\) or \([b / 2 + 1..b]\), then we set $v$ to be the first, second or third internal node at the same depth to which $u$ points, respectively, and implement \(\extract (u, i, j)\) as \(\extract (v, i', j')\), where $i'$ and $j'$ are $i$ and $j$ plus the appropriate offset.  Otherwise, \([i..j]\) must be more than a quarter of \([1..b]\); we split \([i..j]\) into subintervals and implement \(\extract (u, i, j)\) with an \extract\ query for each subinterval, as before.

\begin{theorem}
Extracting a substring \(s [f..\ell]\) from the block graph of $s$ takes $\Oh{\log n + \ell - f}$ time.
\end{theorem}

\begin{proof}
Consider the query \(\extract({\mathrm{root}, f, \ell})\) and let $d$ be the first depth at which we split the interval.  Descending to depth $d$ takes a total of $\Oh{d}$ time.  By induction, if we perform a query \(\extract (v, i, j)\) on a node $v$ at depth \(d' > d\), then \(j - i + 1\) is more than a quarter of the block size \(2^{\lceil \log_2 n \rceil - d'}\) at that level.  It follows that we make $\Oh{(\ell - f + 1) / 2^{\log n - d'}}$ calls to \extract\ at depth $d'$, each of which takes $\Oh{1}$ time.  Summing over the depths, we use a total of $\Oh{\log n + \ell - f}$ time.
\qed
\end{proof}

One interesting property of our block graph structure is that, at the cost of
storing a node for every possible block of size \(n / 2^d\) --- i.e.,
storing $\Oh{2^d \log n}$ extra bits --- we can remove the top $d$ levels
and, thus, change the overall space bound to $\Oh{z (\log n - d) \log n +
2^d \log n}$ bits and reduce the access time to $\Oh{\log n - d}$.  For
example, if \(d = \log z\), then we store a total of $\Oh{z \log n \log (n / z)}$
bits and need only $\Oh{\log (n / z)}$ time for access.
If \(d = \log (n / \log^2 n)\), then we store a total of $\Oh{z \log n \log \log
n + n / \log n}$ bits and reduce the access
time to $\Oh{\log \log n}$.

Gonz\'alez and Navarro~\cite{GN07} showed how, by applying
grammar-based compression to a difference-coded suffix array (SA), we can
build a new kind of compressed suffix array that supports access to
\(\mathrm{SA} [i..j]\) in $\Oh{\log n + \ell - f}$ time.  It seems likely that, by using
a modified block graph of the difference-coded suffix array instead of
a grammar, we can improve their access time to $\Oh{\log \log n + \ell - f}$
at the cost of only slightly increasing their space bound.

\section{Accelerated approximate pattern matching} \label{sec:accelerated}

Suppose we are given an uncompressed string $s$ of length $n$, the LZ77 parse~\cite{ZL77} of $s$, a pattern $p$ of length \(m \leq n\) and an edit distance \(k \leq m\).  The primary matches of $p$ are the substrings of $s$ within edit distance $k$ of $p$ whose characters are all within distance \((m + k)\) of phrase boundaries in the parse.  It is not difficult to find all $p$'s primary matches in $\Oh{z \min (m k + m, k^4 + m)}$ time, where $z$ is the number of phrases.  To do this, we extract the substrings all of whose characters are within distance \((m + k)\) of phrase boundaries and apply to them either the sequential approximate pattern-matching algorithm by Landau and Vishkin~\cite{LV89} or the one by Cole and Hariharan~\cite{CH02}.

Once we have found $p$'s primary matches, we can use them to find the approximate matches not within distance \((m + k)\) of any phrase boundary, which are called $p$'s secondary matches.  To do this, we process the phrases from left to right, maintaining a sorted list of the approximate matches we have already found.  For each phrase copied from a previous substring \(s [i..j]\), we search in the list to see if there are any approximate matches in \(s [i..j]\) that are not completely contained in \(s [i..i + m + k - 1]\) or \(s [j - m - k + 1..j]\).  If there are, we insert the corresponding secondary matches in our list.  Processing all the phrases takes $\Oh{z + \occ}$ time, where \occ\ is the number of approximate matches to $p$ in $s$.  Notice that finding $p$'s secondary matches does not require access to $s$.

As noted in Section~\ref{sec:intro}, Bille et al.~\cite{BLRSSW11} showed how, given a straight-line program for $s$ with $r$ rules, we can build an $\Oh{r}$-word data structure that allows us to extract any substring \(s [f..\ell]\) in $\Oh{\log n + \ell - f}$ time.  When the straight-line program is built with the best known algorithm for approximately minimizing the number of rules, \(r = \Oh{z \log n}\)~\cite{Ryt03}.  It follows that we can store $s$ in $\Oh{z \log n}$ words such that, given $p$ and $k$, in $\Oh{z (\log n + m)}$ time we can extract all the characters within distance \((m + k)\) of phrase boundaries and, therefore, find all $p$'s approximate matches in $\Oh{z (\min (m k + m, k^4 + m) + \log n) + \occ}$ time.  (Bille et al. themselves gave a bound of $\Oh{r (\min (m k + m, k^4 + m) + \log n) + \occ}$ but, since even the smallest straight-line program for $s$ has at least $z$ rules~\cite{Ryt03}, the one we state is slightly stronger.)

The key to supporting approximate pattern matching in the block graph is the addition of {\em bookmarks},
which will allow us to quickly extract certain regions of the underlying string.
To add a bookmark to a character \(s [i]\), for each block size $b$ in the block graph, we store pointers to the two nodes whose blocks of size \(2 b\) completely contain the first occurrence of the substrings \(s [i - b + 1..i]\) and \(s [i..i + b - 1]\), and those occurrences' offsets in the blocks.  Thus, storing a bookmark takes $\Oh{\log n}$ words.  To extract a substring that touches \(s [i]\), we extract, separately, the parts of the substring to the left and right of \(s [i]\).  Without loss of generality, we assume the part \(s [i..j]\) to the right is longer and consider only how to extract it.  We first find the smallest block size \(b \geq j - i + 1\), then follow the pointer to the node whose block of size \(2 b\) contains the first occurrence \(s [i..i + b - 1]\).  Since that node has height $\Oh{\log (j - i + 1)}$, we can extract \(s [i..j]\) in $\Oh{j - i + 1}$ time.

\begin{lemma} \label{lem:bookmark}
Extracting a substring \(s [f..\ell]\) that touches a bookmark takes $\Oh{\ell - f}$ time.
\end{lemma}

Inserting a bookmark to each phrase boundary in the LZ77 parse takes $\Oh{z \log n}$ words and allows us, given $m$ and $k$, to extract the characters within distance \((m + k)\) of phrase boundaries in a total of $\Oh{z m}$ time. Combined with the approach described above for finding secondary occurrences, we have our main result.

\begin{theorem} \label{thm:matching}
Let $s$ be a string of length $n$ whose LZ77 parse consists of $z$ phrases.  We can store $s$ in $\Oh{z \log n}$ words such that, given a pattern $p$ of length \(m \leq n\) and an edit distance \(k \leq m\), we can find all \occ\ substrings of $s$ within edit distance $k$ of $p$ in $\Oh{z \min (m k + m, k^4 + m) + \occ}$ time.
\end{theorem}

Note that, in the above theorem, the time to find all $p$'s approximate matches is the same as if we were
keeping $s$ uncompressed, as in the approach described at the start of this section.

We note in passing that we can combine our results with those of Kreft and Navarro~\cite{KN11} to obtain a new worst-case upper bound for LZ77-based indexing.  Specifically, replacing their data structures for access to the string by a block graph with a bookmark at each phrase boundary, and replacing two of their other data structures by faster (and larger, but still $\Oh{z \log^2 n}$ bits) data structures, we can store $s$ in $\Oh{z \log^2 n}$ bits such that, given a pattern $p$ of length $m$, we can find all occurrences of $p$ in $s$ in $\Oh{m^2 + (m + \occ) \log \log z}$ time.  Their index is practical but potentially larger and slower in the worst case.

\section{Efficient representation of block graphs}

We now describe an implementation of block graphs which is efficient in
practice. The main idea is to represent the shape of the graph (the internal
nodes and their pointers) using bitvectors and operations from succinct data
structures, and to carefully allocate space for the leaf nodes depending on
their distance from the root. Below we make use of two familiar operations
for bitvectors: $rank$ and $select$. Given a bitvector $B$, a position $i$,
and a type of bit $b$ (either 0 or 1), $rank_b(B,i)$ returns the number of
occurrences of $b$ before position $i$ in $B$ and $select_b(B,i)$ returns the position of the
$i$th $b$ in $B$. Efficient data structures supporting these operations
have been extensively studied (see, e.g.~\cite{OS07,RRR07}).

Each level of the block graph consists of a number of nodes,
either internal nodes, or leaves. Let $B_d$ be a bitvector which says
whether the $i$th node (from the left) at depth $d$ is a leaf, $B_d[i] = 0$, or an
internal node $B_d[i] = 1$. We define another bitvector $R_d$, where
$R_d[i] = 1$ if and only if $B_d[i] = 1$ and $B_d[i + 1] = 1$ for $i <
n - 1$. That is, we mark a 1 bit for each instance of two adjacent
internal nodes in $B_d$, otherwise $R_d[i] = 0$. Let $L_d$ be an array
that holds leaf nodes at depth $d$. The structure of a leaf node is
discussed below. Finally, let $T$ be the concatenation of the textual
representation (ie. the corresponding substrings) of all internal nodes
at the truncated depth.

\paragraph{Navigating the block graph.} The main operation is to
traverse from an internal node to one of its three children. Say we
are currently at the $j$th internal node at depth $d$ of the block graph
--- that is, we are at $B_d[i]$, where $i = select_1(B_d,j)$. Each
internal node has three children. If these children were independent
then locating the left child of the current node would be simply
three times the node's position on its level, that is $3j =
3\cdot rank_1(B_d,i)$. However, in a block graph adjacent internal nodes
share exactly one child, so we correct for this by subtracting the
number of adjacent internal nodes at this depth prior to the current
node --- this is given by $rank_1(R_d,i)$. To find the position
corresponding to the left child of a node in $B_{d+1}$ we compute

$$\mbox{leftchild}(B_d,i) = 3rank_1(B_d,i) - rank_1(R_d,i)$$

Given the address of the left child it is easy to find the center or
right child by adding 1 or 2 respectively to the result of
$\mbox{leftchild}$. If $B_d[i] = 0$ then we are at a leaf node. Intuitively,
to access its leaf information in $L_d$ we call $L_d[rank_0(B_d,i)]$.
Once we reach the truncated depth to access the text of an internal
node we compute its offset in $T$, $T[(rank_1(B_d,i) *
  truncated\ length)]$.

\paragraph{Leaf nodes.} In a block graph leaves point to internal
nodes. For each leaf we store two values, the position of the
destination node on the current level, and an offset in the
destination node pointing to the beginning of the leaf block. Note
that we do not need to store the depth of the destination node. It is,
by definition, on the level above the leaf, and we know this by
keeping keep track of the depth during each step in a traversal. To
improve compression we store leaf positions and offsets in two
separate arrays. At depth $d$ there are no more than \(2^{d + 1} - 1\)
possible nodes, so we can store each position in \(\log (2^{d + 1} -
1)\) bits.  Given that the length of a node at depth $d$ is \(b =
2^{\lceil \log n \rceil - d}\) and leaf nodes point to an internal
node on the level above, we store each offset in \(\log (2^{\lceil \log
  n \rceil - d - 1})\) bits.

\section{Experiments} \label{sec:experiments}

We have developed an implementation of block graphs\footnote{Available at {\tt http://www.github.com/choobin/block-graph}} and tested it on the
real-world texts of the Pizza-Chili Repetitive Corpus\footnote{\tt http://pizzachili.dcc.uchile.cl/repcorpus.html},
a standard testbed for data structures designed for repetitive strings.

We compared compression acheived by the block graph to the LZ-End data structure by Kreft and Navarro~\cite{KN10}, and to the general-purpose compressors {\sf gzip} and {\sf 7zip}; the results are shown in Table~\ref{tab:sizes}.  We used {\sf gzip} and {\sf 7zip} with the settings {\tt -9} and {\tt -t7z -m0=lzma -mx=9 -mfb=64 -md=32m -ms=on}, respectively, while LZ-End was executed with its default settings.  Throughout our experiments all block graphs were truncated such that the smallest blocks each took 4 bytes. Note that {\sf gzip} and {\sf 7zip} provide compression only, not random access, and are included as reference points for acheivable compression.

We then compared how quickly block graphs and LZ-End support extracting substrings of various lengths; the results are shown in Figure~\ref{fig:speeds3}.  Each run of extractions was performed across 10,000 randomly-generated queries.  Experiments were conducted on an Intel Core i7-2600 3.4 GHz processor with 8GB of main memory, running Linux 3.3.4; code was compiled with GCC version 4.7.0 targeting x86\_64 with full optimizations.  Caches were dropped between runs with {\tt sync \&\& echo 1 > /proc/sys/vm/drop\_caches}.

Although {\sf 7zip} achieves much better compression block graphs achieve better compression than {\sf gzip} except on the {\sf Escherichia Coli} and {\sf influenza} files.  Most importantly, our experiments show that block graphs generally achieve compression comparable to that achieved by LZ-End while supporting significantly faster substring extraction.

\begin{table}[t!]
\caption{Size in bytes of repetitive corpus files encoded with ASCII, {\sf gzip}, {\sf 7zip}, LZ-End and block graphs.}
\label{tab:sizes}
\centering
\begin{tabular}{lrrrrr}
{\bf Collection} & {\bf ASCII} & {\bf gzip} & {\bf 7zip} & {\bf LZ-End} & {\bf Block graph} \\
\\
Escherichia Coli & 112,689,515 &  31,535,023 & 6,147,962 & 49,106,638 & 49,716,456 \\
cere             & 461,286,644 & 120,834,282 & 6,077,972 & 41,342,784 & 57,689,376 \\
coreutils        & 205,281,778 &  49,920,838 & 3,999,812 & 35,863,520 & 47,795,692 \\
einstein.en.txt  & 467,626,544 & 163,664,285 &   323,779 &  2,247,204 &  3,969,392 \\
influenza        & 154,808,555 &  10,636,899 & 2,111,974 & 21,507,089 & 33,171,036 \\
kernel           & 257,961,616 &  69,396,104 & 2,087,006 & 19,347,734 & 24,045,332 \\
para             & 429,265,758 & 116,073,220 & 8,117,573 & 57,415,176 & 72,393,196 \\
world leaders    &  46,968,181 &   8,287,665 &   606,438 &  4,525,317 &  7,321,720 \\
\end{tabular}
\end{table}

%

\begin{figure}
\resizebox{100ex}{!}
{\begin{tabular}{cc}
\raisebox{0ex}[100ex][0ex]{\includegraphics{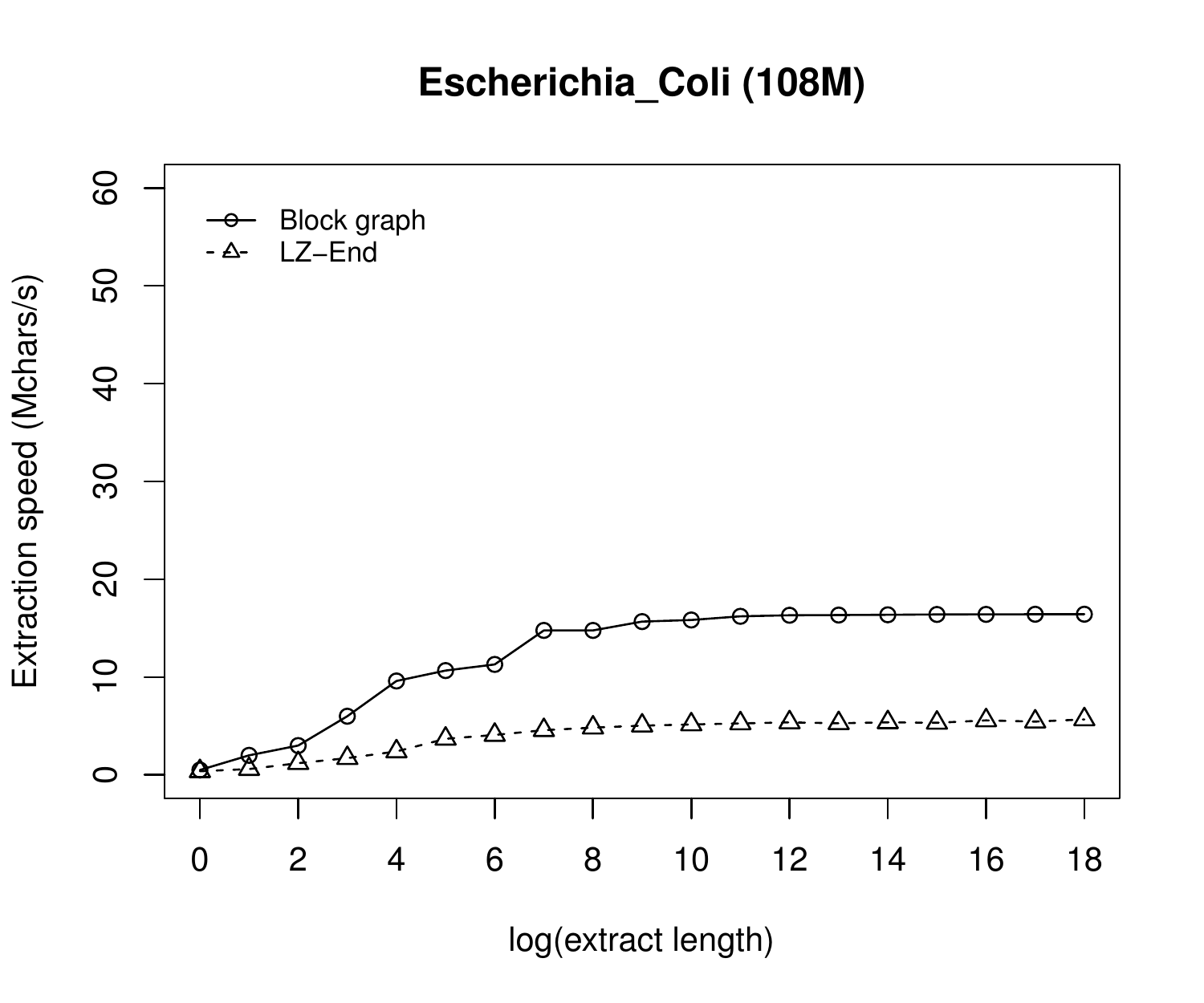}}&
\raisebox{0ex}[100ex][0ex]{\includegraphics{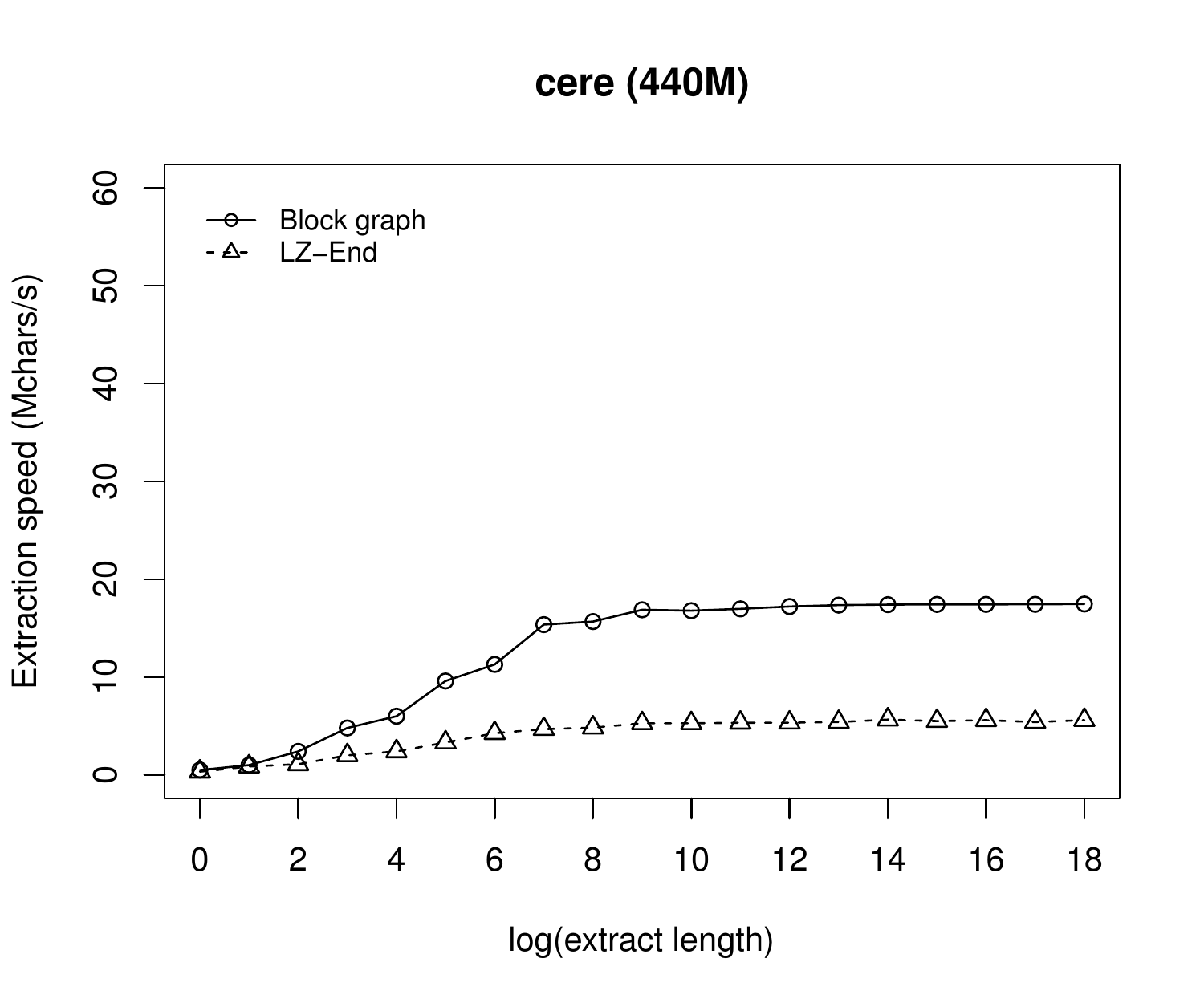}}\\
\raisebox{0ex}[100ex][0ex]{\includegraphics{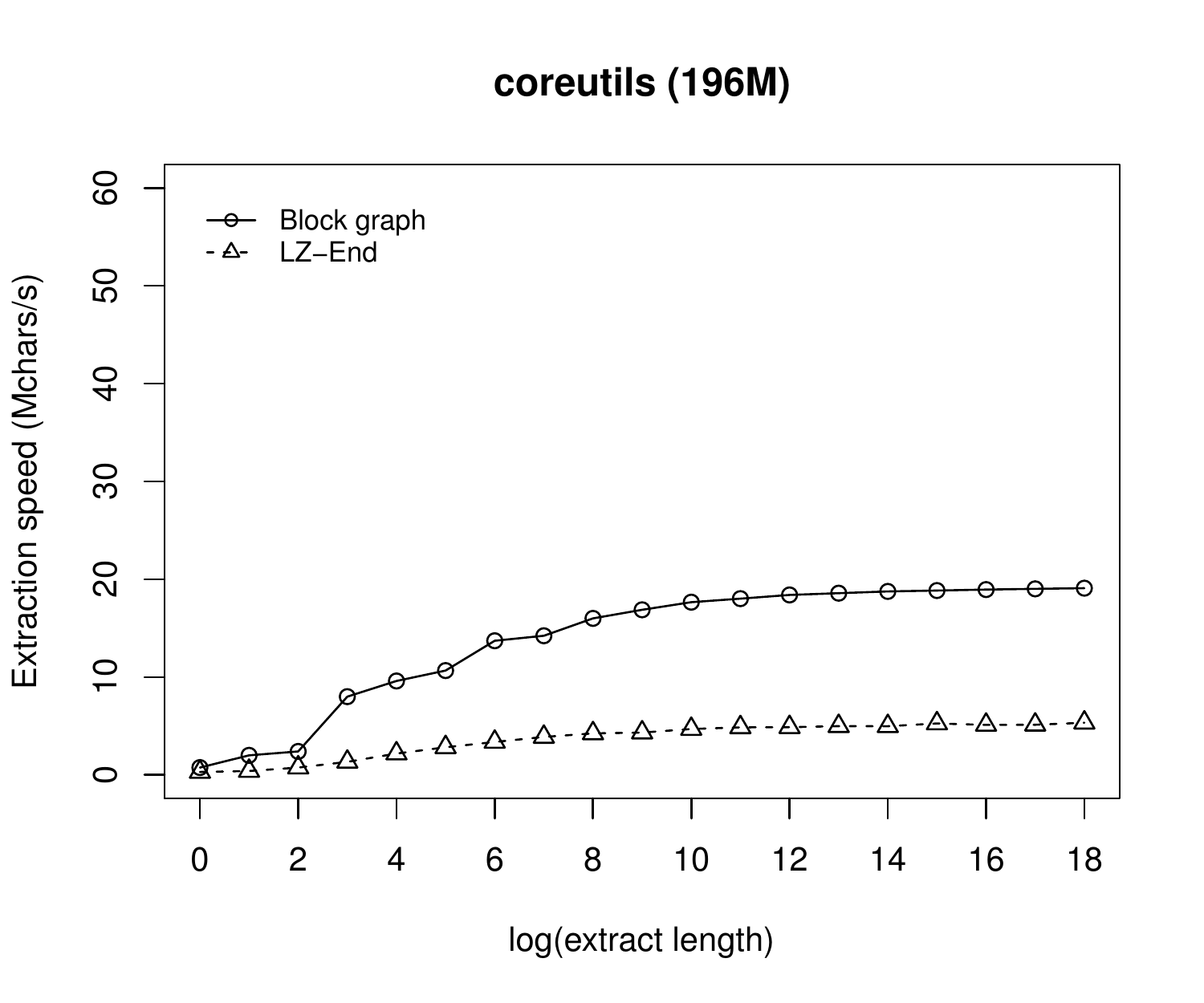}}&
\raisebox{0ex}[100ex][0ex]{\includegraphics{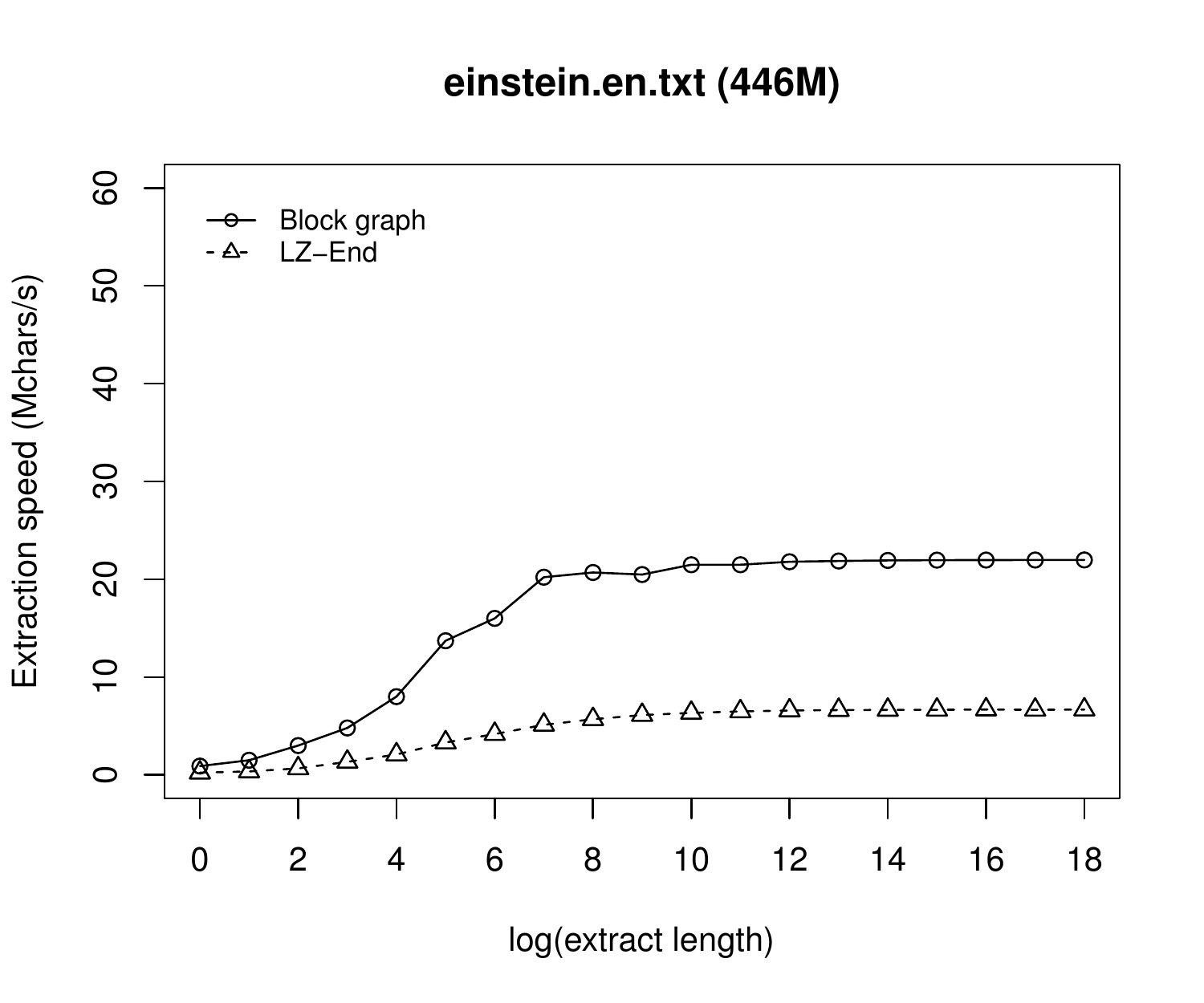}}\\
\raisebox{0ex}[100ex][0ex]{\includegraphics{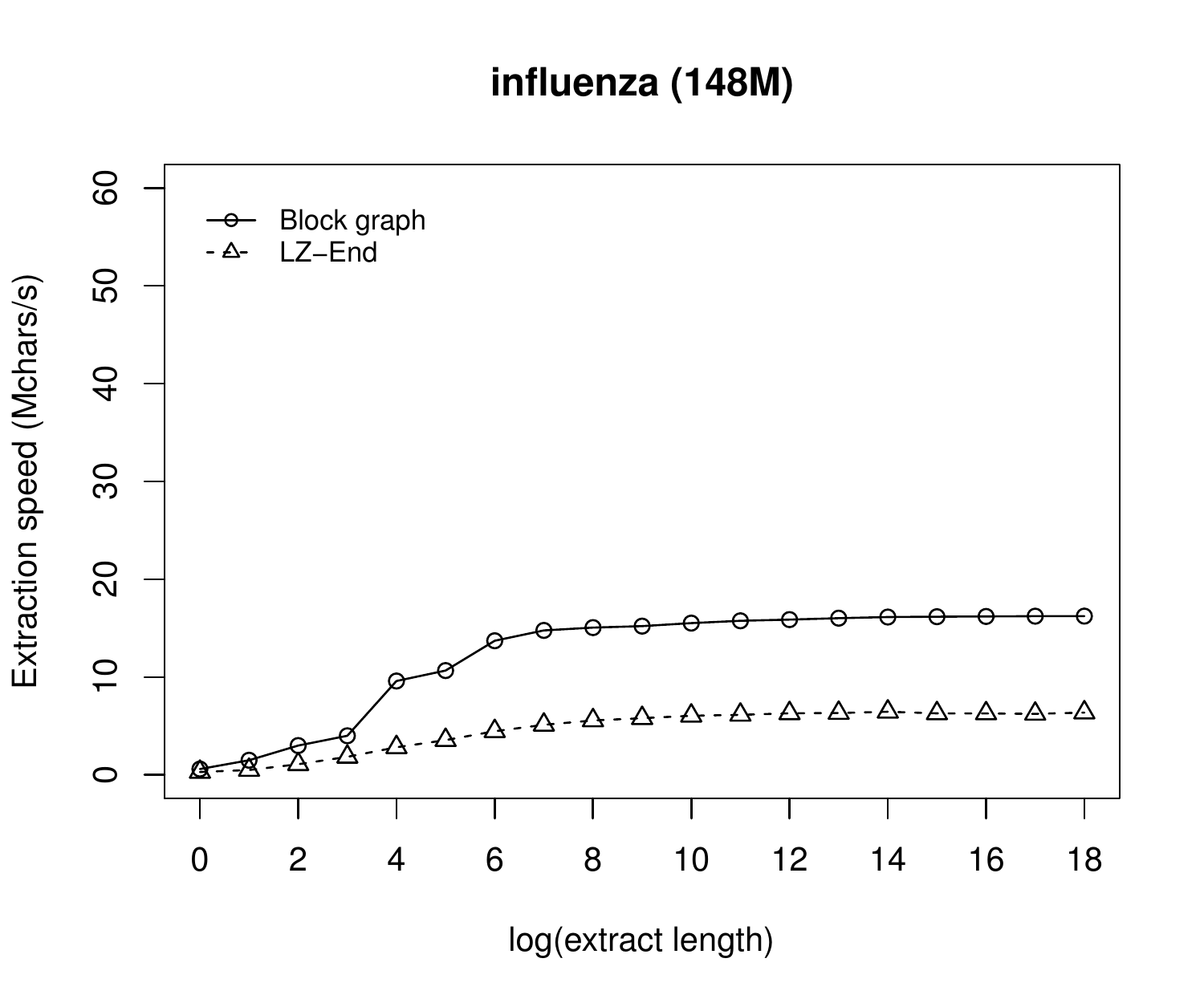}}&
\raisebox{0ex}[100ex][0ex]{\includegraphics{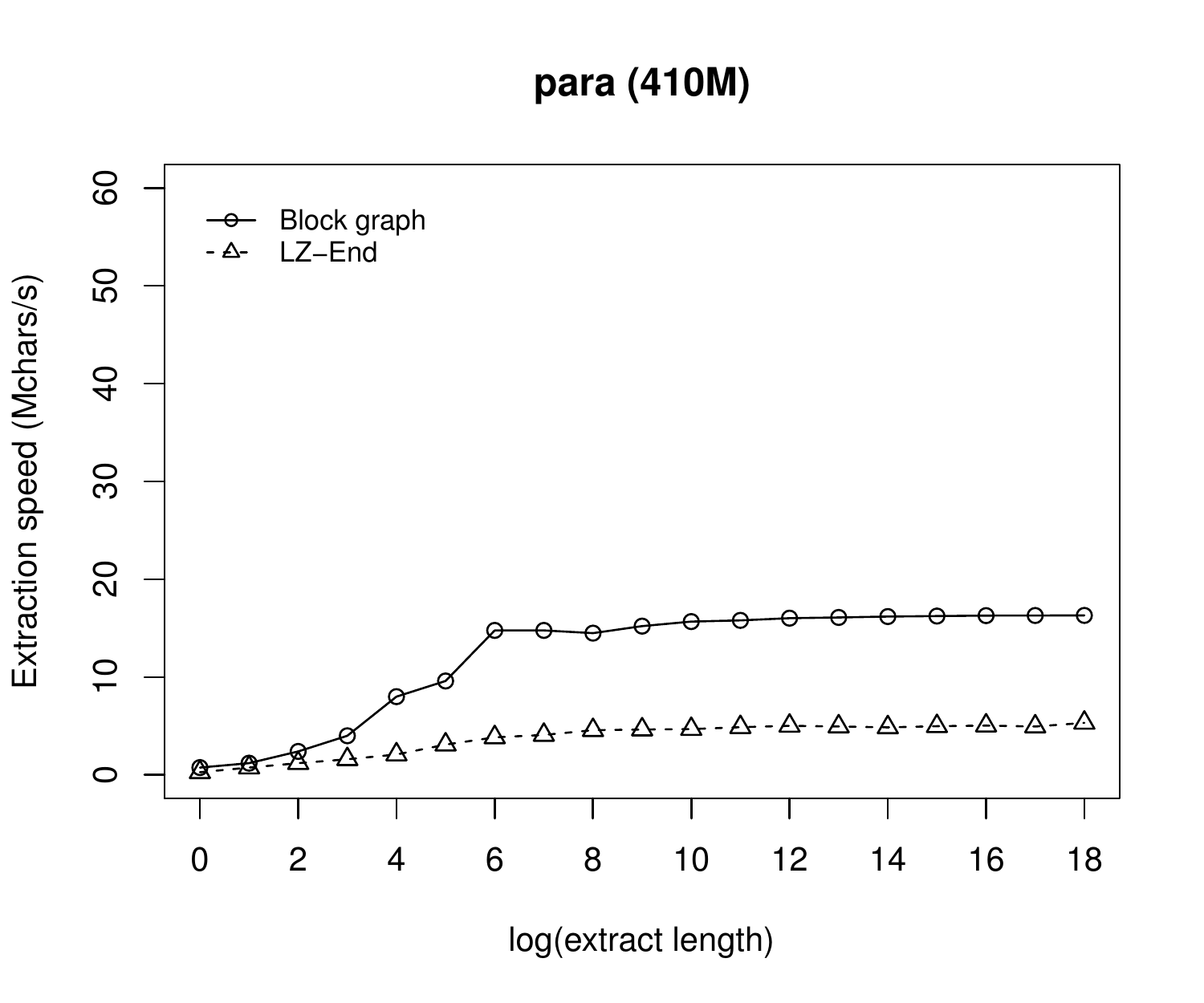}}\\
\raisebox{0ex}[100ex][0ex]{\includegraphics{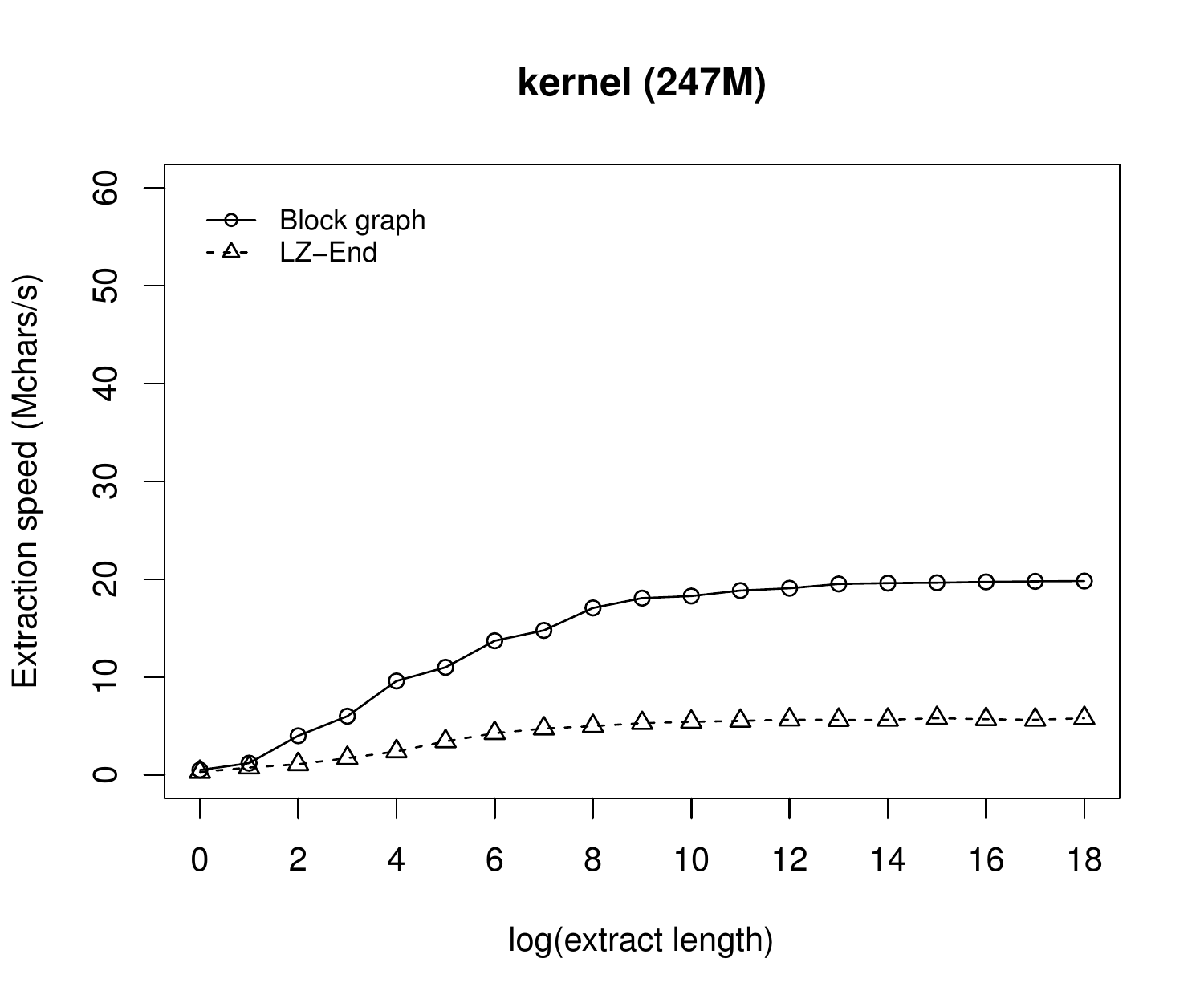}}&
\raisebox{0ex}[100ex][0ex]{\includegraphics{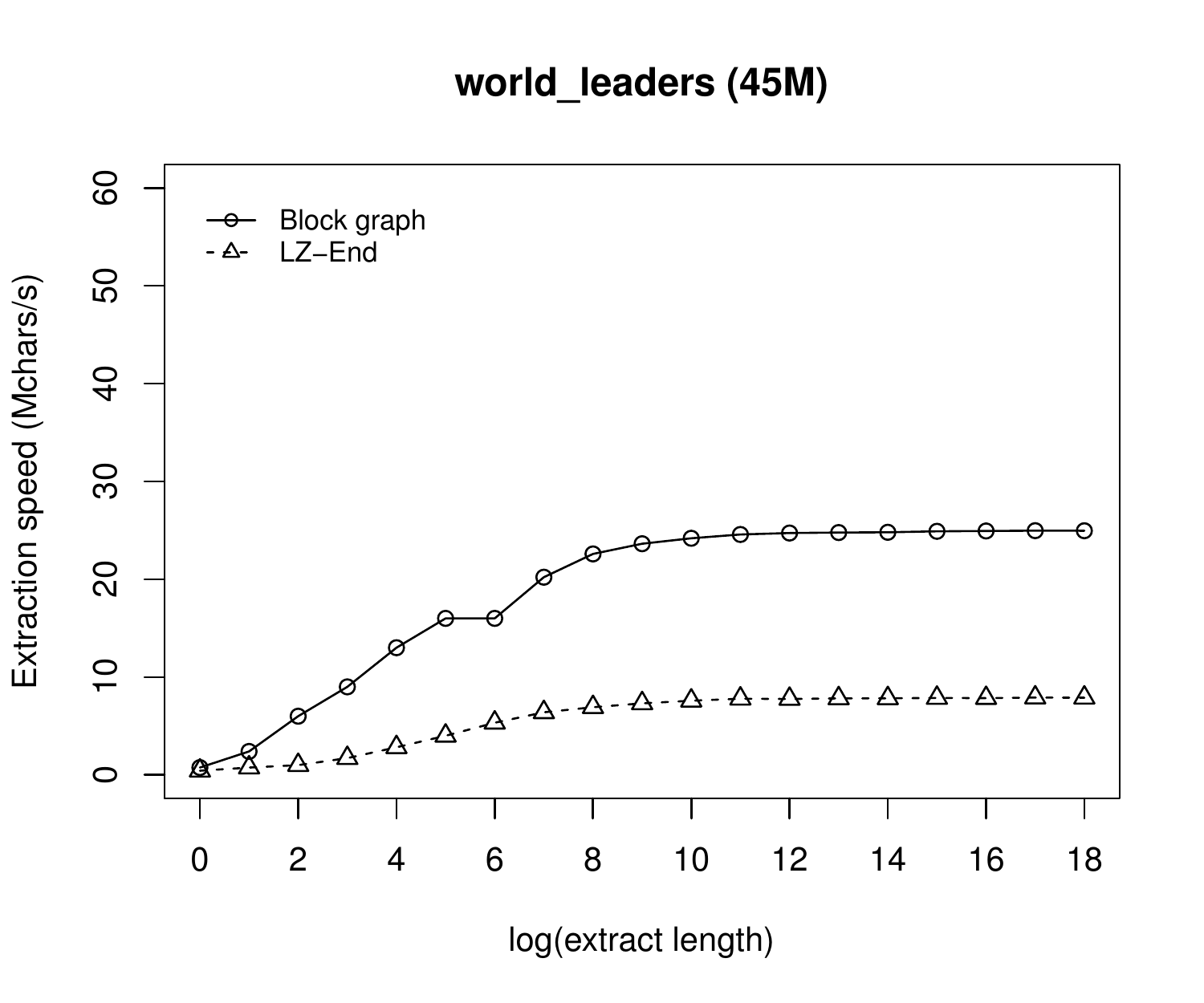}}
\end{tabular}}
\caption{Random access and extraction speeds. Times are averaged over 10,000 random substring extractions.}
\label{fig:speeds3}
\end{figure}

\section{Conclusions} \label{sec:conclusions}

Efficient storage and retrieval of highly repetitive strings, and approximate pattern matching in them, are important tools in bioinformatics and will become even more important as genomic databases grow.  In this paper we have presented a new data structure, the {\em block graph}, that stores highly repetitive strings in compressed space, supports random access in reasonable time and supports extraction from pre-specified points much faster.  Our analysis and experiments show that the block graph is competitive both in theory and in practice.

\section*{Acknowledgments}

Many thanks to Francisco Claude, Juha K\"arkk\"ainen, Sebastian Kreft, Gonzalo Navarro, Jorma Tarhio and Alexandru Tomescu, for helpful discussions.

\bibliographystyle{spbasic}
\bibliography{blocks}

\end{document}